\documentclass[conference]{IEEEtran}
\IEEEoverridecommandlockouts

\newcommand{\longer}[1]{#1}
\newcommand{\shorter}[1]{}

\usepackage{cite}
\usepackage{latexsym}
\usepackage{xspace}
\usepackage{amssymb}
\usepackage{stmaryrd}
\usepackage{hyperref}
\usepackage{url}
\usepackage{graphicx}
\input{xy}
\xyoption{all}
\usepackage{calc}
\usepackage{tikz}
\usepackage{color}

\newtheorem{Theorem}{Theorem}
\newtheorem{Lemma}{Lemma}
\newtheorem{Corollary}{Corollary}
\newtheorem{Proposition}{Proposition}
\newtheorem{Claim}{Claim}

\newtheorem{Definition}{Definition}
\newtheorem{Example}{Example}
\newtheorem{Remark}{Remark}

\newenvironment{definition}
{\smallskip\begin{Definition}}{\end{Definition}\smallskip}

\newenvironment{example}{\smallskip\begin{Example}}{\end{Example}\smallskip}

\newenvironment{theorem}
{\smallskip
	\begin{Theorem}
		\begin{em}}
		{\end{em}
	\end{Theorem}
	\smallskip}

\newenvironment{lemma}
{\smallskip
	\begin{Lemma}
		\begin{em}}
		{\end{em}
	\end{Lemma}
	\smallskip}

\newenvironment{corollary}
{\smallskip
	\begin{Corollary}
		\begin{em}}
		{\end{em}
	\end{Corollary}}

\newenvironment{proposition}
{\smallskip
	\begin{Proposition}
		\begin{em}}
		{\end{em}
	\end{Proposition}
	\smallskip}

\newenvironment{myitemize}{\begin{itemize}}{\end{itemize}}
\newcommand{\myitem}{\item}

\def\eqref#1{(\ref{#1})}

\def\tuple#1{\langle#1\rangle}

\newcommand{\E}{\exists}

\newcommand{\mZ}{\mathcal{Z}}
\newcommand{\mS}{\mathbb{S}}
\newcommand{\mSp}{{\mathbb{S}'}}

\newcommand{\myend}{\mbox{}\hfill{\scriptsize$\blacksquare$}}

\newcommand{\comment}[1]{}

\newcommand{\fand}{\varotimes}

\newcommand{\fto}{\Rightarrow}

\newcommand{\fdPDL}{$\mathit{fPDL}_\triangle$\xspace}

\newcommand{\fpdPDL}{$\mathit{fPDL}^{\mathit{pos}}_\triangle$\xspace}
\newcommand{\fpdK}{$\mathit{fK}^{\mathit{pos}}_{\!\triangle}$\xspace}
\newcommand{\fpdKwxs}{$\mathit{fK}^{\mathit{pos}}_{\!\triangle}$}

\newcommand{\fpedPDL}{$\mathit{fPDL}^{\mathit{pos_\Diamond}}_\triangle$\xspace}
\newcommand{\fpudPDL}{$\mathit{fPDL}^{\mathit{pos_\Box}}_\triangle$\xspace}

\newcommand{\fedPDL}{$\mathit{fPDL}^\E_\triangle$\xspace}
\newcommand{\fedKz}{$\mathit{fK}^{\E,0}_{\!\triangle}$\xspace}
\newcommand{\fedKzwxs}{$\mathit{fK}^{\E,0}_{\!\triangle}$}

\newcommand{\FLTS}{FLTS\xspace}
\newcommand{\FLTSs}{FLTSs\xspace}

\newcommand{\SP}{\Sigma_L}
\newcommand{\SA}{\Sigma_A}

\makeatletter
\def\ps@IEEEtitlepagestyle{%
	\def\@oddfoot{\mycopyrightnotice}%
	\def\@evenfoot{}%
}
\def\mycopyrightnotice{%
	{\footnotesize This is a revised and extended version of a FUZZ-IEEE'2021 paper.\hfill} 
	\gdef\mycopyrightnotice{}
}

\begin{document}
\sloppy

\title{Characterizing Crisp Simulations and Crisp Directed Simulations between Fuzzy Labeled Transition Systems by Using Fuzzy Modal Logics}

\author{\IEEEauthorblockN{Linh Anh Nguyen}
\IEEEauthorblockA{\textit{Institute of Informatics} \\
\textit{University of Warsaw}\\
Warsaw, Poland\\
nguyen@mimuw.edu.pl}
\and
\IEEEauthorblockN{Ngoc-Thanh Nguyen}
\IEEEauthorblockA{\textit{Faculty of Computer Science and Management} \\
\textit{Wroclaw University of Science and Technology}\\
Wroclaw, Poland\\
Ngoc-Thanh.Nguyen@pwr.edu.pl}
}

\maketitle

\begin{abstract}
We formulate and prove logical characterizations of crisp simulations and crisp directed simulations between fuzzy labeled transition systems with respect to fuzzy modal logics that use a general t-norm-based semantics. The considered logics are fragments of the fuzzy propositional dynamic logic with the Baaz projection operator. The logical characterizations concern preservation of positive existential (respectively, positive) modal formulas under crisp simulations (respectively, crisp directed simulations), as well as the Hennessy-Milner property of such simulations. 
\end{abstract}



\section{Introduction}
\label{section:intro}

Like bisimulation, simulation is a well-known notion for comparing observational behaviors of automata and labeled transition systems (LTSs) \cite{Park81,DBLP:journals/fac/He89}. Given states $x$ and $x'$, we say that $x'$ {\em simulates} $x$ if the label of $x$ is a subset of the label of $x'$ and, for every $\sigma$-transition from $x$ to a state $y$, there exists a $\sigma$-transition from $x'$ to a state $y'$ that simulates $y$, where $\sigma$ is any action. Such defined simulations preserve positive existential modal formulas, which are modal formulas without implication, negation and universal modal operators. That is, if $x'$ simulates $x$, then it satisfies all positive existential modal formulas that $x$ does. Conversely, given image-finite LTSs $\mS$ and $\mSp$, if a state $x'$ of $\mSp$ satisfies all positive existential modal formulas that a state $x$ of $\mS$ does, then the pair $\tuple{x,x'}$ belongs to the largest simulation between $\mS$ and $\mSp$ (cf.\ \cite{thesis-Rijke,BRV2001}). This is called the Hennessy-Milner property of simulation. 

Directed simulation is a stronger notion than simulation. 
A~state $x'$ {\em directedly simulates} a state $x$ if: 
\begin{itemize}
\item the label of $x$ is a subset of the label of $x'$; 
\item for every $\sigma$-transition from $x$ to a state $y$, there exists a $\sigma$-transition from $x'$ to a state $y'$ that directedly simulates~$y$; 
\item for every $\sigma$-transition from $x'$ to a state $y'$, there exists a $\sigma$-transition from $x$ to a state $y$ that is directedly simulated by~$y'$.
\end{itemize}
Thus, directed simulation requires both the ``forward'' and ``backward'' conditions of bisimulation and is weaker than bisimulation only in that $x$ and $x'$ are not required to have the same label as in the case of bisimulation.
Directed simulation was introduced and studied by Kurtonina and de Rijke~\cite{KurtoninaR97} for modal logic. 
They proved that directed simulation characterizes the class of positive modal formulas, like simulation characterizes the class of positive existential modal formulas. Positive modal formulas are modal formulas without implication and negation. 
Directed simulation has also been formulated and studied for description logics~\cite{BSDL-P-LOGCOM}. 

For fuzzy structures like fuzzy automata and {\em fuzzy labeled transition systems} (\FLTSs), researchers have studied both crisp simulations~\cite{DamljanovicCI14,DBLP:journals/ijar/PanLC15,DBLP:journals/fss/WuD16,Nguyen-TFS2019,abs-2012-01845} and fuzzy simulations~\cite{CiricIDB12,DBLP:journals/ijar/PanC0C14,IgnjatovicCS15,DBLP:journals/ijar/PanLC15,TFS2020}. Crisp/fuzzy bisimulations have also been studied for fuzzy structures by a considerable number of researchers \cite{CaoCK11,CiricIDB12,EleftheriouKN12,CaoSWC13,DamljanovicCI14,ai/FanL14,Fan15,IgnjatovicCS15,DBLP:journals/fss/WuD16,DBLP:journals/ijar/WuCHC18,DBLP:journals/fss/WuCBD18,jBSfDL2,minimization-by-fBS,abs-2010-15671,abs-2101-12349}. However, as far as we know, only the work~\cite{abs-2012-01845} has concerned crisp directed simulations for fuzzy structures. It only deals with computational aspects. 

The current paper concerns logical characterizations of crisp simulations and crisp directed simulations for \FLTSs. 
As related works on logical characterizations of crisp/fuzzy bisimulations or simulations for fuzzy structures, the notable ones are the papers \cite{EleftheriouKN12,ai/FanL14,Fan15,jBSfDL2,abs-2101-12349} on fuzzy bisimulations, \cite{ai/FanL14,DBLP:journals/fss/WuD16,DBLP:journals/ijar/WuCHC18,DBLP:journals/fss/WuCBD18,jBSfDL2} on crisp bisimulations, \cite{DBLP:journals/ijar/PanC0C14,DBLP:journals/ijar/PanLC15} on fuzzy simulations, and \cite{DBLP:journals/ijar/PanLC15,DBLP:journals/fss/WuD16,Nguyen-TFS2019} on crisp simulations. 
We discuss below only the last three works. 

In~\cite{DBLP:journals/ijar/PanLC15} Pan {\em et at.}\ studied simulations for quantitative transition systems (QTSs), which are a variant of \FLTSs without labels for states. The authors provided logical characterizations of cut-based crisp simulations between finite QTSs w.r.t.\ an existential cut-based crisp Hennessy-Milner logic. A fuzzy threshold, used as the cut for the fuzzy equality relation between actions, is a parameter for both the crisp simulations and the crisp Hennessy-Milner logic under consideration. The main results of~\cite{DBLP:journals/ijar/PanLC15} are formulated only for the case when the underlying residuated lattice is a finite Heyting algebra. 

In~\cite{DBLP:journals/fss/WuD16} Wu and Deng provided a logical characterization of crisp simulations for \FLTSs w.r.t.\ a crisp Hennessy-Milner logic, which uses values from the interval $[0,1]$ as thresholds for modal operators. States of \FLTSs considered in~\cite{DBLP:journals/fss/WuD16} are not labeled. The logical characterization of crisp simulations provided in~\cite{DBLP:journals/fss/WuD16} is the Hennessy-Milner property formulated w.r.t.\ a crisp modal logic with a minimal set of constructors, namely with $\top$, $\land$ and $\tuple{a}_p$, where $a$ is an action and $p \in [0,1]$. 

In~\cite{Nguyen-TFS2019} Nguyen introduced and studied cut-based crisp simulations between fuzzy interpretations in fuzzy description logics under the Zadeh semantics. He provided results on preservation of information under such simulations and the Hennessy-Milner property of such simulations w.r.t.\ fuzzy description logics under the Zadeh semantics. 

As seen from the above discussion, logical characterizations of crisp simulations studied in~\cite{DBLP:journals/ijar/PanLC15,DBLP:journals/fss/WuD16} are formulated w.r.t.\ {\em crisp} modal logics, whereas logical characterizations of crisp simulations studied in~\cite{Nguyen-TFS2019} are formulated w.r.t.\ fuzzy description logics under the {\em Zadeh} semantics. In addition, crisp simulations studied in~\cite{DBLP:journals/ijar/PanLC15,Nguyen-TFS2019} are cut-based and, indeed, a form of fuzzy simulations. There was the lack of logical characterizations of crisp simulations between fuzzy structures w.r.t.\ {\em fuzzy} modal/description logics that use a {\em residuated lattice} or a {\em t-norm-based} semantics. Furthermore, logical characterizations of crisp/fuzzy directed simulations between fuzzy structures have not yet been studied. 

In this paper, we formulate and prove logical characterizations of crisp simulations and crisp directed simulations between \FLTSs w.r.t.\ {\em fuzzy} modal logics that use a general {\em t-norm-based} semantics. The considered logics are fragments of the fuzzy propositional dynamic logic with the Baaz projection operator (\fdPDL). The logical characterizations concern preservation of positive existential (resp.\ positive) modal formulas under crisp simulations (resp.\ crisp directed simulations), as well as the Hennessy-Milner property of such simulations. 

The rest of this paper is structured as follows. Section~\ref{section: prel} contains definitions about fuzzy sets, fuzzy operators, \FLTSs and \fdPDL. In Section~\ref{section: cs} (resp.~\ref{section: cds}), we define crisp simulations (resp.\ crisp directed simulations) between \FLTSs, formulate and prove their logical characterizations w.r.t.\ the positive existential (resp.\ positive) fragments of \fdPDL. Conclusions are given in Section~\ref{section: conc}.


\section{Preliminaries}
\label{section: prel}

\subsection{Fuzzy Sets and Fuzzy Operators}

A {\em fuzzy subset} of a set $X$ is a function \mbox{$f: X \to [0,1]$}. 
Given a fuzzy subset $f$ of $X$, $f(x)$ for $x \in X$ means the fuzzy degree of that $x$ belongs to the subset. 
For \mbox{$\{x_1,\ldots,x_n\} \subseteq X$} and \mbox{$\{a_1,\ldots,a_n\} \subset [0,1]$}, we write $\{x_1:a_1$, \ldots, \mbox{$x_n:a_n\}$} to denote the fuzzy subset $f$ of $X$ such that $f(x_i) = a_i$ for $1 \leq i \leq n$ and $f(x) = 0$ for $x \in X \setminus \{x_1,\ldots,x_n\}$. 
Given fuzzy subsets $f$ and $g$ of a set $X$, we write $f \leq g$ to denote that $f(x) \leq g(x)$ for all $x \in X$. 
A fuzzy subset of $X \times Y$ is called a {\em fuzzy relation} between $X$ and~$Y$. 

An operator \mbox{$\fand: [0,1] \times [0,1] \to [0,1]$} is called a {\em t-norm} if it is commutative and associative, has~1 as the neutral element, and is increasing w.r.t.\ both the arguments. If $L = [0,1]$, $\fand$ is a left-continuous t-norm and \mbox{$\fto\,: [0,1] \times [0,1] \to [0,1]$} is the operator defined by \mbox{$(x \fto y) =$} $\sup \{ z \mid z \fand x \leq y\}$, then $\fto$ is called the {\em residuum} of~$\fand$. 

\longer{The three well-known t-norms named after G\"odel, {\L}ukasiewicz and product are specified below:
\begin{eqnarray*}
x \fand_G y & \ =\ & \min\{x,y\}, \\
x \fand_{\scriptsize{\L}} y & \ =\ & \max\{0, x+y-1\}, \\
x \fand_P y & \ =\ & x \cdot y.
\end{eqnarray*}
The corresponding residua are specified below:
\begin{eqnarray*}
x \fto_G y & \ =\  & \textrm{if $x \leq y$ then 1 else $y$}, \\
x \fto_{\scriptsize{\L}} y & \ =\ & \min\{1, 1 - x + y\}, \\ 
x \fto_P y & \ =\ & \textrm{if $x \leq y$ then 1 else $y/x$}.
\end{eqnarray*}
}

From now on, let $\fand$ be an arbitrary left-continuous t-norm and $\fto$ be its residuum. 
\shorter{For example, we may have \mbox{$x \fand y = \min\{x,y\}$} and $(x \fto y)$ = (if $x \leq y$ then 1 else $y$), using the G\"odel family of fuzzy operators.} 
It is known that $\fto$ is decreasing w.r.t.\ the first argument and increasing w.r.t.\ the second argument. 
Furthermore, $(x \fto y) = 1$ iff $x \leq y$.  

Let $\triangle: [0,1] \times [0,1] \to [0,1]$ be the operator defined by $\triangle x =$ (if $x = 1$ then 1 else 0).

\subsection{Fuzzy Labeled Transition Systems} 

Let $\SA$ be a non-empty set of actions and $\SP$ a non-empty set of state labels. 
We use $\varrho$ to denote an element of $\SA$ and $p$ to denote an element of~$\SP$.

An \FLTS is a triple $\mS = \tuple{S,\delta,L}$, where $S$ is a non-empty set of {\em states}, \mbox{$\delta: S \times \SA \times S \to [0,1]$} is called the {\em transition function}, and \mbox{$L: S \to (\SP \to [0,1])$} is called the {\em state labeling function}. For $x,y \in S$, $\varrho \in \SA$ and $p \in \SP$, $\delta(x,\varrho,y)$ means the fuzzy degree of that there is a transition of the action~$\varrho$ from the state~$x$ to the state~$y$, whereas the fuzzy subset $L(x)$ of $\SP$ is the label of $x$ and $L(x)(p)$ means the fuzzy degree of that $p$ belongs to the label of~$x$. 

An \FLTS $\mS = \tuple{S,\delta,L}$ is {\em image-finite} if, for every $x \in S$ and $\varrho \in \SA$, the set $\{y \mid \delta(x,\varrho,y) > 0\}$ is finite. It is {\em finite} if $S$, $\SA$ and $\SP$ are finite.


\subsection{Fuzzy PDL with the Baaz Projection Operator}

We use $\tuple{\SA,\SP}$ as the signature for the logical languages considered in this article. By \fdPDL we denote the fuzzy propositional dynamic logic with the Baaz projection operator. {\em Programs} and {\em formulas} of \fdPDL over the signature $\tuple{\SA,\SP}$  are defined as follows:
\begin{itemize}
	\item actions from $\SA$ are programs of \fdPDL, 
	\item if $\alpha$ and $\beta$ are programs of \fdPDL, then $\alpha \circ \beta$, $\alpha \cup \beta$ and $\alpha^*$ are also programs of \fdPDL, 
	\item if $\varphi$ is a formula of \fdPDL, then $\varphi?$ is a program of \fdPDL, 
	\item values from the interval $[0,1]$ and propositions from $\SP$ are formulas of \fdPDL, 
	\item if $\varphi$ and $\psi$ are formulas of \fdPDL and $\alpha$ is a program of \fdPDL, then $\triangle \varphi$, $\varphi \land \psi$, $\varphi \lor \psi$, $\varphi \to \psi$, $[\alpha]\varphi$ and $\tuple{\alpha}\varphi$ are formulas of \fdPDL.  
\end{itemize} 

Note that we ignore negation, as a formula $\lnot\varphi$ is usually defined to be $\varphi \to 0$. 

Given a finite set $\Gamma = \{\varphi_1,\ldots,\varphi_n\}$ of formulas with $n \geq 0$, we define $\bigwedge\Gamma = \varphi_1 \land\ldots\land \varphi_n \land 1$. 

\begin{definition}\label{def: BHDJA}
Treating an \FLTS $\mS = \tuple{S,\delta,L}$ as a fuzzy Kripke model, a program $\alpha$ is interpreted in $\mS$ as a fuzzy relation $\alpha^\mS: S \times S \to [0,1]$, whereas a formula $\varphi$ is interpreted in $\mS$ as a fuzzy subset \mbox{$\varphi^\mS: S \to [0,1]$}. The functions $\alpha^\mS$ and $\varphi^\mS$ are specified as follows. 
\begin{eqnarray*}
\!\!\!\!\!\!\!\!\!\! \varrho^\mS(x,y) & = & \delta(x,\varrho,y) \\
\!\!\!\!\!\!\!\!\!\! (\varphi?)^\mS(x,y) & = & \textrm{(if $x = y$ then $\varphi^\mS(x)$ else 0)} \\
\!\!\!\!\!\!\!\!\!\! (\alpha \circ \beta)^\mS(x,y) & = & \sup\{\alpha^\mS(x,z) \fand \beta^\mS(z,y) \mid z \in S \} \\
\!\!\!\!\!\!\!\!\!\! (\alpha \cup \beta)^\mS(x,y) & = & \max\{\alpha^\mS(x,y),\beta^\mS(x,y)\} \\
\!\!\!\!\!\!\!\!\!\! (\alpha^*)^\mS(x,y) & = & \sup \{\textstyle\bigotimes\{\alpha^\mS(x_i,x_{i+1}) \mid 0 \leq i < n\} \\
\!\!\!\!\!\!\!\!\!\! && \qquad \mid n \geq 0,\ x_0,\ldots,x_n \in S,\\
\!\!\!\!\!\!\!\!\!\! && \qquad\quad x_0 = x,\ x_n = y\} \\
\!\!\!\!\!\!\!\!\!\! a^\mS(x) & = & a \\
\!\!\!\!\!\!\!\!\!\! p^\mS(x) & = & L(x,p) \\
\!\!\!\!\!\!\!\!\!\! (\triangle\varphi)^\mS(x) & = & \textrm{(if $\varphi^\mS(x) = 1$ then 1 else 0)} \\
\!\!\!\!\!\!\!\!\!\! (\varphi \land \psi)^\mS(x) & = & \min\{\varphi^\mS(x),\psi^\mS(x)\} \\
\!\!\!\!\!\!\!\!\!\! (\varphi \lor \psi)^\mS(x) & = & \max\{\varphi^\mS(x),\psi^\mS(x)\} \\
\!\!\!\!\!\!\!\!\!\! (\varphi \to \psi)^\mS(x) & = & (\varphi^\mS(x) \fto \psi^\mS(x)) \\
\!\!\!\!\!\!\!\!\!\! ([\alpha]\varphi)^\mS(x) & = & \inf \{\alpha^\mS(x,y) \fto \varphi^\mS(y) \mid y \in S\} \\ 
\!\!\!\!\!\!\!\!\!\! (\tuple{\alpha}\varphi)^\mS(x) & = & \sup \{\alpha^\mS(x,y) \fand \varphi^\mS(y) \mid y \in S\}.
\end{eqnarray*}

\vspace{-3.2ex}

\myend
\end{definition}


\section{Crisp Simulations between \FLTSs and Their Logical Characterizations}
\label{section: cs}

In this section, we first recall crisp simulations between \FLTSs, then define the positive existential fragments \fedPDL and \fedKz of \fdPDL, and finally formulate and prove logical characterizations of crisp simulations between \FLTSs w.r.t.\ these positive existential fragments of \fdPDL. 

\subsection{Crisp Simulations between \FLTSs}

This subsection is a reformulation of the corresponding one of~\cite{abs-2012-01845} (which concerns fuzzy graphs).

Let $\mS = \tuple{S, \delta, L}$ and $\mS' = \tuple{S', \delta', L'}$ be \FLTSs. 
A binary relation $Z \subseteq S \times S'$ is called a {\em (crisp) simulation} between $\mS$ and $\mS'$ if the following conditions hold for every $x,y \in S$, $x' \in S'$ and $\varrho \in \SA$:
\begin{eqnarray}
&&\!\!\!\!\!\!\!\!\!\!\!\!\!\!\! Z(x,x') \to (L(x) \leq L'(x')) \label{eq: CS 1} \\
&&\!\!\!\!\!\!\!\!\!\!\!\!\!\!\! [Z(x,x') \land (\delta(x,\varrho,y) > 0)] \to \nonumber \\ 
&&\!\!\!\!\!\!\!\!\!\!\!\!\!\!\! \qquad \E y' \in S'\,[(\delta(x,\varrho,y) \leq \delta'(x',\varrho,y')) \land Z(y,y')]. \label{eq: CS 2}
\end{eqnarray}
Here, $\to$ and $\land$ denote the usual crisp logical connectives. Thus, the above conditions mean that: 
\begin{itemize}
	\item[\eqref{eq: CS 1}] if $Z(x,x')$ holds, then $L(x) \leq L'(x')$;
	
	\smallskip
	
	\item[\eqref{eq: CS 2}] if $Z(x,x')$ holds and $\delta(x,\varrho,y) > 0$, then there exists $y' \in S'$ such that $\delta(x,\varrho,y) \leq \delta'(x',\varrho,y')$ and $Z(y,y')$ holds.
\end{itemize}

\begin{figure}[h]
	\begin{tikzpicture}[->,>=stealth,auto]
	\node (G) {$\mS$};
	\node (Gp) [node distance=4.5cm, right of=G] {$\,\mS'$};
	\node (uG) [node distance=1.9cm, below of=G] {};
	\node (a) [node distance=1.5cm, left of=uG] {$u_1: 0.7$};
	\node (virtual) [node distance=1.3cm, left of=a] {};
	\node (b) [node distance=1.5cm, right of=uG] {$u_2: 0.7$};
	\node (c) [node distance=3cm, below of=a] {$u_3: 0.6$};
	\node (d) [node distance=3cm, below of=b] {$u_4: 0.8$};
	\node (e) [node distance=1.9cm, below of=Gp] {$v_1: 0.7$};
	\node (f) [node distance=3cm, below of=e] {$v_2: 0.8$};
	\draw (a) to node{0.6} (b);
	\draw (b) to node[left,yshift=1mm]{0.5} (c);
	\draw (c) to node[below]{0.4} (d);
	\draw (b) edge [bend right=15] node[left]{0.6} (d);
	\draw (d) edge [bend right=15] node[right]{0.5} (b);
	\draw (e) edge [loop above,in=60,out=120,looseness=10] node{0.5} (e);
	\draw (e) edge [bend right=15] node[left]{0.6} (f);
	\draw (f) edge [bend right=15] node[right]{0.5} (e);
	\end{tikzpicture}
	\caption{An illustration for Example~\ref{example: HGDSK}.\label{fig: HGDSK}}
\end{figure}

\begin{example}\label{example: HGDSK}
Let $\SA = \{\varrho\}$, $\SP = \{p\}$ and let $\mS$ and $\mS'$ be the \FLTSs specified below and depicted in Fig.~\ref{fig: HGDSK}.
	\begin{myitemize}
		\myitem $\mS = \tuple{S$, $\delta$, $L}$, where $S = \{u_1$, $u_2$, $u_3$, $u_4\}$, 
		$L(u_1)(p) = 0.7$, $L(u_2)(p) = 0.7$, $L(u_3)(p) = 0.6$, $L(u_4)(p) = 0.8$ and 
		$\delta = \{\tuple{u_1,\varrho,u_2}:0.6$, $\tuple{u_2,\varrho,u_3}:0.5$, $\tuple{u_2,\varrho,u_4}:0.6$, $\tuple{u_3,\varrho,u_4}:0.4$, $\tuple{u_4,\varrho,u_2}:0.5\}$.
		
		\myitem $\mS' = \tuple{S'$, $\delta'$, $L'}$, where $S' = \{v_1$, $v_2\}$,
		$L'(v_1)(p) = 0.7$, $L'(v_2)(p) = 0.8$ and 
		$\delta' = \{\tuple{e,\varrho,e}:0.5$, $\tuple{e,\varrho,f}:0.6$, $\tuple{f,\varrho,e}:0.5\}$.
	\end{myitemize}
	
It can be checked that $\{\tuple{u_2,v_1}$, $\tuple{u_3,v_1}$, $\tuple{u_4,v_2}\}$ is the largest simulation between $\mS$ and $\mS'$. 
\myend
\end{example}

\longer{
A ({\em crisp}) {\em auto-simulation} of $\mS$ is a simulation between $\mS$ and itself. 

\begin{proposition}\label{prop: HGDFJ}
Let $\mS$, $\mS'$ and $\mS''$ be \FLTSs and let $\mS = \tuple{S, \delta, L}$. 
	\begin{enumerate}
		\item The relation $Z = \{\tuple{x,x} \mid x \in S\}$ is an auto-simulation of $\mS$.
		\item If $Z_1$ is a simulation between $\mS$ and $\mS'$, and $Z_2$ is a simulation between $\mS'$ and $\mS''$, then $Z_1 \circ Z_2$ is a simulation between $\mS$ and~$\mS''$.
		\item If $\mZ$ is a set of simulations between $\mS$ and $\mS'$, then $\bigcup\mZ$ is also a~simulation between $\mS$ and $\mS'$.
	\end{enumerate}   
\end{proposition}

The proof of this proposition is straightforward. 

\begin{corollary}
The largest simulation between two arbitrary \FLTSs exists.
The largest auto-simulation of a \FLTS is a pre-order.
\end{corollary}
}

\subsection{The Positive Existential Fragment of \fdPDL}

By \fedPDL we denote the largest sublanguage of \fdPDL that disallow the formula constructor $[\alpha]\varphi$ and allows implication ($\to$) only in formulas of the form $a \to \varphi$ with $a \in [0,1]$. We call \fedPDL the {\em positive existential fragment of \fdPDL}. 

By \fedKz we denote the largest sublanguage of \fedPDL that disallow the program constructors ($\alpha \circ \beta$, $\alpha \cup \beta$, $\alpha^*$ and $\varphi?$), the disjunction operator ($\lor$) and the constructor $a$ (with $a \in [0,1]$). That is, only actions from $\SA$ are programs of \fedKz, and formulas of \fedKz are of the form $p$, $\triangle \varphi$, $\varphi \land \psi$, $a \to \varphi$ or $\tuple{\varrho}\varphi$, where $p \in \SP$, $a \in [0,1]$, $\varrho \in \SA$, and $\varphi$ and $\psi$ are formulas of \fedKz. 

An \FLTS $\mS = \tuple{S,\delta,L}$ is said to be {\em witnessed} w.r.t.\ \fedPDL if, for every formula $\varphi$ (resp.\ program $\alpha$) of \fedPDL and every $x,y \in S$, if the definition of $\varphi^\mS(x)$ (resp.\ $\alpha^\mS(x,y)$) in Definition~\ref{def: BHDJA} uses supremum, then the set under the supremum has the biggest element if it is non-empty. 

The notion of whether an \FLTS $\mS = \tuple{S,\delta,L}$ is {\em witnessed} w.r.t.\ \fedKz is defined analogously.  

Observe that if an \FLTS $\mS = \tuple{S,\delta,L}$ is finite, then it is witnessed w.r.t.\ \fedPDL and \fedKz. If $\mS$ is image-finite, then it is witnessed w.r.t.\ \fedKz.

\subsection{Logical Characterizations of Simulations between FLTSs}

A formula $\varphi$ is said to be {\em preserved} under simulations between \FLTSs if, for every \FLTSs $\mS = \tuple{S, \delta, L}$ and $\mS' = \tuple{S', \delta', L'}$ that are witnessed w.r.t.\ \fedPDL, for every simulation $Z$ between them, and for every $x \in S$ and $x' \in S'$, if $Z(x,x')$ holds, then $\varphi^\mS(x) \leq \varphi^\mSp(x')$.

\begin{theorem}\label{theorem: pres-sim}
All formulas of \fedPDL are preserved under simulations between \FLTSs.
\end{theorem}

This theorem follows immediately from the first assertion of the following lemma.

\begin{lemma} \label{lemma: pres-sim}
Let $\mS = \tuple{S, \delta, L}$ and $\mS' = \tuple{S', \delta', L'}$ be \FLTSs witnessed w.r.t.\ \fedPDL and $Z$ be a simulation between them. Then, the following assertions hold for every $x,y \in S$, $x' \in S'$, every formula $\varphi$ and every program $\alpha$ of \fedPDL, where $\to$ and $\land$ are the usual crisp logical connectives:
\begin{eqnarray}
&&\!\!\!\!\!\!\!\!\!\!\!\!\!\!\! Z(x,x') \to (\varphi^\mS(x) \leq \varphi^\mSp(x')) \label{eq: CSx 1} \\[0.5ex]
&&\!\!\!\!\!\!\!\!\!\!\!\!\!\!\! [Z(x,x') \,\land\, (\alpha^\mS(x,y) \!>\! 0)] \to \nonumber\\
&&\!\!\!\!\!\!\!\!\!\!\!\!\!\!\! \qquad \E y' \in S'\, [(\alpha^\mS(x,y) \!\leq\! \alpha^\mSp(x',y')) \,\land\, Z(y,y')]. \label{eq: CSx 2}
\end{eqnarray}
\end{lemma}

\begin{proof}
We prove this lemma by induction on the structure of $\varphi$ and $\alpha$. 
	
Consider the assertion~\eqref{eq: CSx 1}. 
Assume that $Z(x,x')$ holds. We need to show that $\varphi^\mS(x) \leq \varphi^\mSp(x')$. 
The cases when $\varphi$ is a constant $a \in [0,1]$ or a proposition $p \in \SP$ are trivial. 
The cases when $\varphi$ is of the form $\triangle\psi$, $\psi \land \xi$ or $\psi \lor \xi$ are also straightforward by using the induction assumptions about $\psi$ and $\xi$ and the definition of $\varphi^\mS(x)$ and $\varphi^\mSp(x')$. 
The remaining cases are considered below.
\begin{myitemize}
\myitem Case $\varphi = a \to \psi$: By the induction assumption, $\psi^\mS(x) \leq \psi^\mSp(x')$. It is known that the residuum $\fto$ of every continuous t-norm $\fand$ is increasing w.r.t.\ the second argument. Hence, $\varphi^\mS(x) \leq \varphi^\mSp(x')$. 

\myitem Case $\varphi = \tuple{\alpha}\psi$: For a contradiction, assume that $\varphi^\mS(x) > \varphi^\mSp(x')$. 
Since $\mS$ is witnessed w.r.t.\ \fedPDL, there exists $y \in S$ such that $\varphi^\mS(x) = \alpha^\mS(x,y) \fand \psi^\mS(y)$. Since $\varphi^\mS(x) > \varphi^\mSp(x')$, it follows that $\varphi^\mS(x) > 0$ and therefore $\alpha^\mS(x,y) > 0$ (since $0 \fand a \leq 0 \fand 1 = 0$ for all $a \in [0,1]$). By the induction assumption of \eqref{eq: CSx 2}, there exists $y' \in S'$ such that $\alpha^\mS(x,y) \leq \alpha^\mSp(x',y')$ and $Z(y,y')$ holds. Since $Z(y,y')$ holds, by the induction assumption, $\psi^\mS(y) \leq \psi^\mSp(y')$. Since $\fand$ is increasing w.r.t.\ both the arguments, it follows that 
\[ \alpha^\mS(x,y) \fand \psi^\mS(y) \leq \alpha^\mSp(x',y') \fand \psi^\mSp(y'). \]
This contradicts the assumption $\varphi^\mS(x) > \varphi^\mSp(x')$.
\end{myitemize}

Consider the assertion~\eqref{eq: CSx 2}. The case when $\alpha$ is an action $\varrho \in \SA$ follows from Condition~\eqref{eq: CS 2}. The other cases are considered below.

\begin{myitemize}
\myitem Case $\alpha = \beta_1 \circ \beta_2$: Suppose that $Z(x,x')$ holds and $\alpha^\mS(x,y) > 0$. Thus, 
\[ \!\!\!\!\!\alpha^\mS(x,y) = \sup\{\beta_1^\mS(x,z) \fand \beta_2^\mS(z,y) \mid z \in S \} > 0. \]
Since $\mS$ is witnessed w.r.t.\ \fedPDL, there exists $z \in S$ such that 
\[ \alpha^\mS(x,y) = \beta_1^\mS(x,z) \fand \beta_2^\mS(z,y) > 0. \] 
Since $a \fand 0 = 0 \fand a = 0$ for all $a \in [0,1]$, we must have that $\beta_1^\mS(x,z) > 0$ and $\beta_2^\mS(z,y) > 0$. Since $Z(x,x')$ holds, by the induction assumption, there exists $z' \in S'$ such that $Z(z,z')$ holds and $\beta_1^\mS(x,z) \leq \beta_1^\mSp(x',z')$. Since $Z(z,z')$ holds and $\beta_2^\mS(z,y) > 0$, by the induction assumption, there exists $y' \in S'$ such that $Z(y,y')$ holds and $\beta_2^\mS(z,y) \leq \beta_2^\mSp(z',y')$. Since $\fand$ is increasing, it follows that 
\[
\beta_1^\mS(x,z) \fand \beta_2^\mS(z,y) \leq \beta_1^\mSp(x',z') \fand \beta_2^\mSp(z',y').
\]
Therefore, $\alpha^\mS(x,y) \leq \alpha^\mSp(x',y')$ and the induction hypothesis~\eqref{eq: CSx 2} holds.  

\myitem Case $\alpha = \beta_1 \cup \beta_2$: Suppose that $Z(x,x')$ holds and \mbox{$\alpha^\mS(x,y) > 0$}. Thus, $\max\{\beta_1^\mS(x,y)$, $\beta_2^\mS(x,y)\} > 0$. W.l.o.g.\ we assume that \mbox{$\beta_1^\mS(x,y) > 0$}. Since $Z(x,x')$ holds, by the induction assumption, it follows that there exists $y' \in S'$ such that $\beta_1^\mS(x,y) \leq \beta_1^\mSp(x',y')$ and $Z(y,y')$ holds. Therefore, $\alpha^\mS(x,y) =$ $\beta_1^\mS(x,y) \leq$ $\beta_1^\mSp(x',y') \leq$ $\alpha^\mSp(x',y')$ and the induction hypothesis \eqref{eq: CSx 2} holds. 

\myitem Case $\alpha = \beta^*$: Suppose that $Z(x,x')$ holds and $\alpha^\mS(x,y) > 0$. If $x = y$, then by taking $y' = x'$, $\alpha^\mSp(x',y') = 1 = \alpha^\mS(x,y)$ and $Z(y,y')$ holds. Assume that $x \neq y$. Since $\mS$ is witnessed w.r.t.\ \fedPDL, there exist $n \geq 1$ and $x_0,\ldots,x_n \in S$ such that $x_0 = x$, $x_n = y$ and $(\beta^*)^\mS(x,y) =$ \mbox{$\beta^\mS(x_0,x_1) \fand\cdots\fand \beta^\mS(x_{n-1},x_n)$}. Since $\alpha^\mS(x,y) > 0$, we must have that $\beta^\mS(x_i,x_{i+1}) > 0$ for all $0 \leq i < n$ (because $a \fand 0 = 0 \fand a = 0$ for all $a \in [0,1]$). Let $x'_0 = x'$. For each $i$ from $0$ to $n-1$, since $Z(x_i,x'_i)$ holds and $\beta^\mS(x_i,x_{i+1}) > 0$, by the induction assumption, there exists $x'_{i+1} \in S'$ such that $\beta^\mS(x_i,x_{i+1}) \leq \beta^\mSp(x'_i,x'_{i+1})$ and $Z(x_{i+1},x'_{i+1})$ holds. Take $y' = x'_n$. Thus, $Z(y,y')$ holds. Since $\fand$ is increasing,\\[0.5ex] 
\mbox{$\qquad \beta^\mS(x_0,x_1) \fand\cdots\fand \beta^\mS(x_{n-1},x_n) \leq$}\\
\mbox{$\qquad\qquad\qquad \beta^\mSp(x'_0,x'_1) \fand\cdots\fand \beta^\mSp(x'_{n-1},x'_n)$}.\\[0.5ex]
Hence, $\alpha^\mS(x,y) \leq \alpha^\mSp(x',y')$ and the induction hypothesis \eqref{eq: CSx 2} holds.

\myitem Case $\alpha = (\psi?)$: Suppose that $Z(x,x')$ holds and $\alpha^\mS(x,y) > 0$. Thus, $x = y$ and $\alpha^\mS(x,y) = \psi^\mS(x)$. Since $Z(x,x')$ holds, by the induction assumption~\eqref{eq: CSx 1}, $\psi^\mS(x) \leq \psi^\mSp(x')$. By choosing $y' = x'$, we have that $\alpha^\mS(x,y) = \psi^\mS(x) \leq \psi^\mSp(x') = \alpha^\mSp(x',y')$ and the induction hypothesis \eqref{eq: CSx 2} holds.

\vspace{-2.8ex}

\end{myitemize}
\end{proof}

The following lemma is a counterpart of Lemma~\ref{lemma: pres-sim} for \fedKz (instead of \fedPDL). 
Its proof can obtained from the proof of the assertion~\eqref{eq: CSx 1} of Lemma~\ref{lemma: pres-sim} by simplification, using~\eqref{eq: CS 2} instead of~\eqref{eq: CSx 2}.

\begin{lemma} \label{lemma: pres-sim 2}
Let $\mS = \tuple{S, \delta, L}$ and $\mS' = \tuple{S', \delta', L'}$ be \FLTSs witnessed w.r.t.\ \fedKz and $Z$ be a simulation between them. Then, for every $x \in S$ and $x' \in S'$, if $Z(x,x')$ holds, then $\varphi^\mS(x) \leq \varphi^\mSp(x')$ for all formulas $\varphi$ of \fedKz.
\end{lemma}


An \FLTS $\mS = \tuple{S, \delta, L}$ is said to be {\em modally saturated} w.r.t.\ \fedKz if, for every $x \in S$, $a \in (0,1]$, $\varrho \in \SA$ and every infinite set $\Gamma$ of formulas of~\fedKz, if for every finite subset $\Psi$ of $\Gamma$ there exists $y \in S$ such that $\varrho^\mS(x,y) \fand (\bigwedge\!\Psi)^\mS(y) \geq a$, then there exists $y \in S$ such that $\varrho^\mS(x,y) \geq a$ and $\varphi^\mS(y) > 0$ for all $\varphi \in \Gamma$. 
The notion of modal saturatedness is a technical extension of image-finiteness. 
\longer{This is confirmed by the following proposition. 

\begin{proposition}
Every image-finite \FLTS is modally saturated w.r.t.\ \fedKz.
\end{proposition}

\begin{proof}
Let $\mS = \tuple{S, \delta, L}$ be an image-finite \FLTS. 
Let $x \in S$, $a \in (0,1]$, $\varrho \in \SA$ and let $\Gamma$ be an infinite set of formulas of~\fedKz. 
We prove that $\mS$ is modally saturated w.r.t.\ \fedKz by contraposition. 
Suppose that, for every $y \in S$, there exists $\varphi_y \in \Gamma$ such that $\varrho^\mS(x,y) < a$ or $\varphi_y^\mS(y) = 0$. 
We need to prove that there exists a finite subset $\Psi$ of $\Gamma$ such that, for every $y \in S$, $\varrho^\mS(x,y) \fand (\bigwedge\!\Psi)^\mS(y) < a$. 
Let $\Psi = \{\varphi_y \mid y \in S$ and $\varrho^\mS(x,y) > 0\}$. Since $\mS$ is image-finite, $\Psi$ is finite. 
For every $y \in S$, since either $\varrho^\mS(x,y) < a$ or $\varphi_y \in \Psi$ and $\varphi_y^\mS(y) = 0$, we must have that $\varrho^\mS(x,y) \fand (\bigwedge\!\Psi)^\mS(y) < a$. This completes the proof. 
\end{proof}
}
\shorter{It can be checked that every image-finite \FLTS is modally saturated w.r.t.\ \fedKz.}

The following theorem states the Hennessy-Milner property of crisp simulations between \FLTSs. 

\begin{theorem} \label{theorem: sim HM}
Let $\mS = \tuple{S, \delta, L}$ and $\mS' = \tuple{S', \delta', L'}$ be \FLTSs witnessed and modally saturated w.r.t.\ \fedKz. 
Then, $Z = \{\tuple{x,x'} \in S \times S' \mid \varphi^\mS(x) \leq \varphi^\mSp(x')$ for all formulas $\varphi$ of \fedKzwxs$\}$ is the largest simulation between $\mS$ and~$\mSp$.
\end{theorem}

\begin{proof}
By Lemma~\ref{lemma: pres-sim 2}, it is sufficient to prove that the considered $Z$ is a simulation between $\mS$ and~$\mSp$. Let $x,y \in S$, $x' \in S'$ and $\varrho \in \SA$. We need to prove Conditions~\eqref{eq: CS 1} and~\eqref{eq: CS 2}. 

Condition~\eqref{eq: CS 1} holds by the definition of~$Z$. 

Consider Condition~\eqref{eq: CS 2} and suppose that $Z(x,x')$ holds and $\delta(x,\varrho,y) = a > 0$. 
Let $Y' = \{y' \in S' \mid \varrho^\mSp(x',y') \geq a\}$. We need to show that there exists $y' \in Y'$ such that $Z(y,y')$ holds. For a contradiction, suppose that, for each $y' \in Y'$, $Z(y,y')$ does not hold, which means there exists a formula $\varphi_{y'}$ of \fedKz such that $\varphi_{y'}^\mS(y) > \varphi_{y'}^\mSp(y')$. For every $y' \in Y'$, let $\psi_{y'} = \triangle(\varphi_{y'}^\mS(y) \to \varphi_{y'})$, which is a formula of \fedKz. Observe that, for every $y' \in Y'$, $\psi_{y'}^\mS(y) = 1$ and $\psi_{y'}^\mSp(y') = 0$. Let $\Gamma = \{\psi_{y'} \mid y' \in Y'\}$. Observe that, for every $y' \in S'$, either $\varrho(x',y') < a$ or there exists $\psi = \psi_{y'} \in \Gamma$ such that $\psi^\mSp(y') = 0$. Since $\mSp$ is modally saturated w.r.t.\ \fedKz, there exists a finite subset $\Psi$ of $\Gamma$ such that, for every $y' \in S'$, $\varrho^\mSp(x',y') \fand (\bigwedge\!\Psi)^\mSp(y') < a$. Let $\varphi = \tuple{\varrho}\bigwedge\!\Psi$. It is a formula of \fedKz. Since $\mSp$ is witnessed w.r.t.\ \fedKz, it follows that $\varphi^\mSp(x') < a$. Since $\psi^\mS(y) = 1$ for all $\psi \in \Psi$, $(\bigwedge\!\Psi)^\mS(y) = 1$ and $\varphi^\mS(x) \geq a$. Thus, $\varphi^\mS(x) > \varphi^\mSp(x')$, which contradicts the assumption that $Z(x,x')$ holds. 
\end{proof}

Let $\mS = \tuple{S, \delta, L}$ and $\mS' = \tuple{S', \delta', L'}$ be \FLTSs and let $x \in S$ and $x' \in S'$. We write $x \lesssim^s x'$ to denote that there exists a simulation $Z$ between $\mS$ and $\mSp$ such that $Z(x,x')$ holds. We also write \mbox{$x \leq^\E x'$} (resp.\ \mbox{$x \leq^{\E,0}_K x'$}) to denote that $\varphi^\mS(x) \leq \varphi^\mSp(x')$ for all formulas $\varphi$ of~\fedPDL (resp.\ \fedKz). The following corollary follows immediately from Theorems~\ref{theorem: sim HM} and~\ref{theorem: pres-sim}. 

\begin{corollary} \label{cor: CS}
Let $\mS = \tuple{S, \delta, L}$ and $\mS' = \tuple{S', \delta', L'}$ be \FLTSs and let $x \in S$ and $x' \in S'$. 
\begin{enumerate}
\item If $\mS$ and $\mSp$ are witnessed and modally saturated w.r.t.~\fedKz, then
\longer{\[ x \lesssim^s x'\ \ \textrm{iff}\ \ x \leq^{\E,0}_K x', \]}
\shorter{$x \lesssim^s x'$ iff $x \leq^{\E,0}_K x'$,}
and therefore whether $x \leq^{\E,0}_K x'$ or not does not depend on the used t-norm~$\fand$. 
	
\item If $\mS$ and $\mSp$ are witnessed w.r.t.~\fedPDL and modally saturated w.r.t.~\fedKz, then
\longer{\[ x \leq^\E x'\ \ \textrm{iff}\ \ x \lesssim^s x'\ \ \textrm{iff}\ \ x \leq^{\E,0}_K x', \]}
\shorter{$x \leq^\E x'$ iff $x \lesssim^s x'$ iff $x \leq^{\E,0}_K x'$,} 
and therefore whether $x \leq^\E x'$ or not does not depend on the used t-norm~$\fand$. 
\end{enumerate}	
\end{corollary}


\section{Crisp Directed Simulations between FLTSs and Their Logical Characterizations}
\label{section: cds}

In this section, we first recall crisp directed simulations between \FLTSs, then define the positive fragments \fpdPDL and \fpdK of \fdPDL, and finally formulate and prove logical characterizations of crisp directed simulations between \FLTSs w.r.t.\ these positive fragments of \fdPDL. 

\subsection{Crisp Directed Simulations between FLTSs}

This subsection is a reformulation of the corresponding one of~\cite{abs-2012-01845} (which concerns fuzzy graphs).

Let $\mS = \tuple{S, \delta, L}$ and $\mS' = \tuple{S', \delta', L'}$ be \FLTSs. 
A binary relation $Z \subseteq S \times S'$ is called a {\em (crisp) directed simulation} between $\mS$ and $\mS'$ if it satisfies Conditions~\eqref{eq: CS 1} and~\eqref{eq: CS 2} (of simulations) and the following one for every $x \in S$, $x', y' \in S'$ and $\varrho \in \SA$, where $\to$ and $\land$ denote the usual crisp logical connectives: 
\begin{eqnarray}
&&\!\!\!\!\!\!\!\!\!\!\!\!\!\!\! [Z(x,x') \land (\delta'(x',\varrho,y') > 0)] \to \nonumber \\ 
&&\!\!\!\!\!\!\!\!\!\!\!\!\!\!\! \qquad \E y \in S\,[(\delta'(x',\varrho,y') \leq \delta(x,\varrho,y)) \land Z(y,y')]. \label{eq: CS 3}
\end{eqnarray}

\begin{figure}[h]
	\begin{tikzpicture}[->,>=stealth,auto]
	\node (G) {$\mS_2$};
	\node (Gp) [node distance=4.5cm, right of=G] {$\,\mS_2'$};
	\node (uG) [node distance=1.9cm, below of=G] {};
	\node (a) [node distance=1.5cm, left of=uG] {$u_1: 0.7$};
	\node (virtual) [node distance=1.3cm, left of=a] {};
	\node (b) [node distance=1.5cm, right of=uG] {$u_2: 0.5$};
	\node (c) [node distance=3cm, below of=a] {$u_3: 0.6$};
	\node (d) [node distance=3cm, below of=b] {$u_4: 0.7$};
	\node (e) [node distance=1.9cm, below of=Gp] {$v_1: 0.7$};
	\node (f) [node distance=3cm, below of=e] {$v_2: 0.8$};
	\draw (a) to node{0.7} (b);
	\draw (b) to node[left,yshift=1mm]{0.5} (c);
	\draw (c) to node[below]{0.6} (d);
	\draw (b) edge [bend right=15] node[left]{0.6} (d);
	\draw (d) edge [bend right=15] node[right]{0.6} (b);
	\draw (e) edge [loop above,in=60,out=120,looseness=10] node{0.6} (e);
	\draw (e) edge [bend right=15] node[left]{0.5} (f);
	\draw (f) edge [bend right=15] node[right]{0.6} (e);
	\end{tikzpicture}
\caption{An illustration for Example~\ref{example: HGDSK 2}.\label{fig: HGDSK 2}}
\end{figure}

\begin{example}\label{example: HGDSK 2}
Let $\SA = \{\varrho\}$ and $\SP = \{p\}$. 
Reconsider the \FLTSs $\mS$ and $\mSp$ specified in Example~\ref{example: HGDSK}. 
It can be checked that $\emptyset$ is the unique directed simulation between $\mS$ and $\mSp$.

Let $\mS_2$ and $\mS_2'$ be the \FLTSs illustrated in Fig.~\ref{fig: HGDSK 2} and specified in a similar way as done for $\mS$ and $\mS'$ in Example~\ref{example: HGDSK}. It is straightforward to show that $Z = \{u_2,u_3,u_4\} \times \{v_1,v_2\}$ is the largest directed simulation between $\mS_2$ and $\mS'_2$. 
\myend
\end{example}

\longer{
A ({\em crisp}) {\em directed auto-simulation} of an \FLTS $\mS$ is a directed simulation between $\mS$ and itself. 

The following proposition is a counterpart of Proposition~\ref{prop: HGDFJ 2}. Its proof is also straightforward.

\begin{proposition}\label{prop: HGDFJ 2}
Let $\mS$, $\mS'$ and $\mS''$ be \FLTSs and let $\mS = \tuple{S, \delta, L}$. 
\begin{enumerate}
\item The relation $Z = \{\tuple{x,x} \mid x \in S\}$ is a directed auto-simulation of $\mS$.
\item If $Z_1$ is a directed simulation between $\mS$ and $\mS'$, and $Z_2$ is a directed simulation between $\mS'$ and $\mS''$, then $Z_1 \circ Z_2$ is a directed simulation between $\mS$ and~$\mS''$.
\item If $\mZ$ is a set of directed simulations between $\mS$ and $\mS'$, then $\bigcup\mZ$ is also a directed simulation between $\mS$ and~$\mS'$.
\end{enumerate}   
\end{proposition}

The proof of this proposition is straightforward. 

\begin{corollary}
The largest directed simulation between two arbitrary \FLTSs exists.
The largest directed auto-simulation of a \FLTS is a pre-order.
\end{corollary}
}

\subsection{The Positive Fragment of \fdPDL}

If we disallow the test operator (?), then the positive fragment of \fdPDL would simply be defined to be the largest fragment of \fdPDL that (disallows the test operator and) allows implication ($\to$) only in formulas of the form $a \to \varphi$ with $a \in [0,1]$. 
Allowing the test operator makes the matter more sophisticated, as shown below (cf.\ \cite{BSDL-P-LOGCOM}). 

{\em Formulas of \fpdPDL} and {\em programs} of \fpedPDL and \fpudPDL are defined inductively as follows: 
\begin{itemize}
	\item actions from $\SA$ are programs of \fpedPDL and \fpudPDL; 
	\item if $\alpha$ and $\beta$ are programs of \fpedPDL (resp.\ \fpudPDL), then $\alpha \circ \beta$, $\alpha \cup \beta$ and $\alpha^*$ are also programs of \fpedPDL (resp.\ \fpudPDL); 
	\item if $\varphi$ is a formula of \fpdPDL, then $\varphi?$ is a program of \fpedPDL and $(\varphi \to a)?$ is a program of \fpudPDL, for $a \in [0,1]$; 
	\item values from the interval $[0,1]$ and propositions from $\SP$ are formulas of \fpdPDL; 
	\item if $\varphi$ and $\psi$ are formulas of \fpdPDL, then 
		\begin{itemize}
		\item $\triangle \varphi$, $\varphi \land \psi$ and $\varphi \lor \psi$ are formulas of \fpdPDL,
		\item if $a \in [0,1]$, then $a \to \varphi$ is a formula of \fpdPDL,
		\item if $\alpha$ is a program of \fpedPDL, then $\tuple{\alpha}\varphi$ is a formula of \fpdPDL,   
		\item if $\alpha$ is a program of \fpudPDL, then $[\alpha]\varphi$ is a formula of \fpdPDL.   
		\end{itemize}
\end{itemize} 

We call \fpdPDL the {\em positive fragment} of \fdPDL. 

By \fpdK we denote the largest sublanguage of \fpdPDL that disallow all the program constructors. That is, only actions from $\SA$ are programs of \fpdK, and formulas of \fpdK are of the form $a$, $p$, $\triangle \varphi$, $\varphi \land \psi$, $\varphi \lor \psi$, $a \to \varphi$, $[\varrho]\varphi$ or $\tuple{\varrho}\varphi$, where $a \in [0,1]$, $p \in \SP$, $\varrho \in \SA$, and $\varphi$ and $\psi$ are formulas of \fpdK. 

Note that \fedPDL is the sublanguage of \fpdPDL that disallows the formula constructor $[\alpha]\varphi$, whereas \fedKz is the sublanguage of \fpdK that disallows the formula constructors $[\alpha]\varphi$, $\varphi \lor \psi$ and $a$ (with $a \in [0,1]$).

An \FLTS $\mS = \tuple{S,\delta,L}$ is said to be {\em witnessed} w.r.t.\ \fpdPDL if:
\begin{myitemize}
\myitem for every formula $\varphi$ of \fpdPDL and every $x \in S$, if the definition of $\varphi^\mS(x)$ in Definition~\ref{def: BHDJA} uses supremum (resp.\ infimum), then the set under the supremum (resp.\ infimum) has the biggest (resp.\ smallest) element if it is non-empty, 
\myitem for every program $\alpha$ of \fpedPDL or \fpudPDL and every $x,y \in S$, if the definition of $\alpha^\mS(x,y)$ in Definition~\ref{def: BHDJA} uses supremum (resp.\ infimum), then the set under the supremum (resp.\ infimum) has the biggest (resp.\ smallest) element if it is non-empty. 
\end{myitemize}

The notion of whether an \FLTS $\mS$ is {\em witnessed} w.r.t.\ \fpdK is defined analogously.  

Observe that if an \FLTS $\mS = \tuple{S,\delta,L}$ is finite, then it is witnessed w.r.t.\ \fpdPDL and \fpdK. If $\mS$ is image-finite, then it is witnessed w.r.t.\ \fpdK.

\subsection{Logical Characterizations of Crisp Directed Simulations between \FLTSs}

A formula $\varphi$ is said to be {\em preserved} under directed simulations between \FLTSs if, for every \FLTSs $\mS = \tuple{S, \delta, L}$ and $\mS' = \tuple{S', \delta', L'}$ that are witnessed w.r.t.\ \fpdPDL, for every directed simulation $Z$ between them, and for every $x \in S$ and $x' \in S'$, if $Z(x,x')$ holds, then $\varphi^\mS(x) \leq \varphi^\mSp(x')$.

\begin{theorem}\label{theorem: pres-dir-sim}
All formulas of \fpdPDL are preserved under directed simulations between \FLTSs.
\end{theorem}

This theorem follows immediately from the first assertion of the following lemma, which is a counterpart of Lemma~\ref{lemma: pres-sim}.

\begin{lemma} \label{lemma: pres-dir-sim}
Let $\mS = \tuple{S, \delta, L}$ and $\mS' = \tuple{S', \delta', L'}$ be \FLTSs witnessed w.r.t.\ \fpdPDL and $Z$ be a directed simulation between them. Then, the following assertions hold for every $x,y \in S$, $x',y' \in S'$, every formula $\varphi$ of \fpdPDL, every program $\alpha$ of \fpedPDL and every program $\gamma$ of \fpudPDL, where $\to$ and $\land$ are the usual crisp logical connectives:
\begin{eqnarray}
&&\!\!\!\!\!\!\!\!\!\!\!\!\!\!\! Z(x,x') \to (\varphi^\mS(x) \leq \varphi^\mSp(x')) \label{eq: CDSx 1} \\[0.5ex]
&&\!\!\!\!\!\!\!\!\!\!\!\!\!\!\! [Z(x,x') \,\land\, (\alpha^\mS(x,y) \!>\! 0)] \to \nonumber\\
&&\!\!\!\!\!\!\!\!\!\!\!\!\!\!\! \qquad \E y' \in S'\, [(\alpha^\mS(x,y) \!\leq\! \alpha^\mSp(x',y')) \,\land\, Z(y,y')] \label{eq: CDSx 2} \\
&&\!\!\!\!\!\!\!\!\!\!\!\!\!\!\! [Z(x,x') \,\land\, (\gamma^\mSp(x',y') \!>\! 0)] \to \nonumber\\
&&\!\!\!\!\!\!\!\!\!\!\!\!\!\!\! \qquad \E y \in S\, [(\gamma^\mSp(x',y') \!\leq\! \gamma^\mS(x,y)) \,\land\, Z(y,y')]. \label{eq: CDSx 3}
\end{eqnarray}
\end{lemma}

\begin{proof}
We prove this lemma by induction analogously as done for Lemma~\ref{lemma: pres-sim}. 

In comparison with the proof of Lemma~\ref{lemma: pres-sim}, for the assertion~\eqref{eq: CDSx 1}, we only need to consider the additional case when $\varphi = [\gamma]\psi$ (and $\gamma$ is a program of \fpudPDL). Consider this case. For a contradiction, suppose that $\varphi^\mS(x) > \varphi^\mSp(x')$. 
Since $\mSp$ is witnessed w.r.t.\ \fpdPDL, there exists $y' \in S'$ such that \mbox{$\varphi^\mSp(x') = (\gamma^\mSp(x',y') \fto \psi^\mSp(y'))$}. Since $\varphi^\mS(x) > \varphi^\mSp(x')$, we have $\varphi^\mSp(x') < 1$, which implies that $\gamma^\mSp(x',y') > 0$. By the induction assumption~\eqref{eq: CDSx 3}, there exists $y \in S$ such that $\gamma^\mSp(x',y') \leq \gamma^\mS(x,y)$
and $Z(y,y')$ holds. Since $Z(y,y')$ holds, by the induction assumption, $\psi^\mS(y) \leq \psi^\mSp(y')$. Since $\fto$ is decreasing w.r.t.\ the first argument and increasing w.r.t.\ the second argument, it follows that 
\[ (\gamma^\mS(x,y) \fto \psi^\mS(y)) \leq (\gamma^\mSp(x',y') \fto \psi^\mSp(y')), \] 
which means $\varphi^\mS(x) \leq \varphi^\mSp(x')$, which contradicts the assumption $\varphi^\mS(x) > \varphi^\mSp(x')$. This completes the proof of the assertion~\eqref{eq: CDSx 1}. 

The proof of the assertion~\eqref{eq: CDSx 2} is obtained from the proof of the assertion~\eqref{eq: CSx 2} of Lemma~\ref{lemma: pres-sim} by replacing the occurrences of \fedPDL, \eqref{eq: CSx 1} and \eqref{eq: CSx 2} with \fpdPDL, \eqref{eq: CDSx 1} and \eqref{eq: CDSx 2}, respectively. 

The proof of the assertion~\eqref{eq: CDSx 3} is dual to the proof of the assertion~\eqref{eq: CDSx 2}. The only special difference is that instead of the case $\alpha = (\psi?)$ we need to consider the case $\gamma = (\psi \to a)?$, with $a \in [0,1]$. Consider this case. 
Suppose that $Z(x,x')$ holds and $\gamma^\mSp(x',y') > 0$. Thus, $x' = y'$ and $\gamma^\mSp(x',y') = (\psi^\mSp(x') \fto a)$. Since $Z(x,x')$ holds, by the induction assumption~\eqref{eq: CSx 1}, $\psi^\mS(x) \leq \psi^\mSp(x')$. Since $\fto$ is decreasing w.r.t.\ the first argument, by choosing $y = x$, we have that 
\[ \gamma^\mSp(x',y') = (\psi^\mSp(x') \fto a) \leq (\psi^\mS(x) \fto a) = \gamma^\mS(x,y), \]
and the induction hypothesis \eqref{eq: CDSx 3} holds.
\end{proof}

The following lemma is a counterpart of Lemma~\ref{lemma: pres-dir-sim} for \fpdK (instead of \fpdPDL). 
Its proof can obtained from the proof of the assertion~\eqref{eq: CDSx 1} of Lemma~\ref{lemma: pres-dir-sim} by simplification, using~\eqref{eq: CS 2} and~\eqref{eq: CS 3} instead of~\eqref{eq: CDSx 2} and~\eqref{eq: CDSx 3}, respectively.

\begin{lemma} \label{lemma: pres-dir-sim 2}
Let $\mS = \tuple{S, \delta, L}$ and $\mS' = \tuple{S', \delta', L'}$ be \FLTSs witnessed w.r.t.\ \fpdK and $Z$ be a directed simulation between them. Then, for every $x \in S$ and $x' \in S'$, if $Z(x,x')$ holds, then $\varphi^\mS(x) \leq \varphi^\mSp(x')$ for all formulas $\varphi$ of \fpdK.
\end{lemma}

The following theorem is a counterpart of Theorem~\ref{theorem: sim HM} devoted to the Hennessy-Milner property of crisp directed simulations between \FLTSs. It is formulated only for image-finite \FLTSs. 

\begin{theorem} \label{theorem: dir-sim HM}
Let $\mS = \tuple{S, \delta, L}$ and $\mS' = \tuple{S', \delta', L'}$ be image-finite \FLTSs. Then, $Z = \{\tuple{x,x'} \in S \times S' \mid \varphi^\mS(x) \leq \varphi^\mSp(x')$ for all formulas $\varphi$ of \fpdKwxs$\}$ is the largest directed simulation between $\mS$ and~$\mSp$.
\end{theorem}

\begin{proof}
Clearly, all image-finite \FLTSs are witnessed w.r.t.\ \fpdK. By Lemma~\ref{lemma: pres-dir-sim 2}, it is sufficient to prove that $Z$ is a directed simulation between $\mS$ and~$\mSp$. Let $x \in S$, $x' \in S'$, $p \in \SP$ and $\varrho \in \SA$. We need to prove Conditions~\eqref{eq: CS 1}, \eqref{eq: CS 2} and \eqref{eq: CS 3}. 
Condition~\eqref{eq: CS 1} clearly holds by the definition of $Z$. 
Condition~\eqref{eq: CS 2} can be proved analogously as done for Theorem~\ref{theorem: sim HM} by replacing ``witnessed w.r.t.~\fedKz'' and ``modally saturated w.r.t.\ \fedKz'' with the assumption that $\mSp$ is image-finite and replacing the remaining occurrence of \fedKz with \fpdK. 

We now prove Condition~\eqref{eq: CS 3}. Let $y' \in S'$ and suppose that $Z(x,x')$ holds and $\varrho^\mSp(x',y') = a > 0$. Let $Y = \{y \in S \mid \varrho^\mS(x,y) \geq a\}$. Since $\mS$ is image-finite, $Y$ is finite. We need to prove that there exists $y \in Y$ such that $Z(y,y')$ holds. For a contradiction, suppose that, for every $y \in Y$, $Z(y,y')$ does not hold, which means there exists a formula $\varphi_y$ of \fpdK such that $\varphi_y^\mS(y) > \varphi_y^\mSp(y')$. For every $y \in Y$, let $\psi_y = \triangle(\varphi_y^\mS(y) \to \varphi_y)$, which is a formula of \fpdK. Let $\Psi = \{\psi_y \mid y \in Y\}$. Observe that, for every $y \in Y$, $\psi_y^\mS(y) = 1$ and $\psi_y^\mSp(y') = 0$. Let $Y_l = \{y \in S \mid$ $0 < \varrho^\mS(x,y) < a\}$ and $a_l = \sup \{\varrho^\mS(x,y) \mid y \in Y_l\}$, where the subscript $l$ stands for ``left''. Note that, if $Y_l = \emptyset$, then $a_l = 0$, else $a_l = \max \{\varrho^\mS(x,y) \mid y \in Y_l\}$ (since $\mS$ is image-finite). In any case, $a_l < a$. Let $a_c = (a_l + a) / 2$. Thus, $a_l < a_c < a$. Let $\varphi = [\varrho](\bigvee\!\Psi \lor a_c)$. It is a formula of \fpdK. 
Since $\psi_y^\mSp(y') = 0$ for all $y \in Y$, $(\bigvee\!\Psi)^\mSp(y') = 0$. Hence, $(\bigvee\!\Psi \lor a_c)^\mSp(y') = a_c$ and $\varphi^\mSp(x') \leq (a \fto a_c)$. It follows that $\varphi^\mSp(x') < 1$. Let's estimate $\varphi^\mS(x)$. For every $y \in Y$, since $\psi_y^\mS(y) = 1$, $(\bigvee\!\Psi \lor a_c)^\mS(y) = 1$. In addition, $(\bigvee\!\Psi \lor a_c)^\mS(y) \geq a_c$ for all $y \in Y_l$. Since $a_l < a_c$, we can conclude that $\varphi^\mS(x) = 1$. This contradicts the facts that $Z(x,x')$ holds and $\varphi^\mSp(x') < 1$. This completes the proof. 
\end{proof}

Let $\mS = \tuple{S, \delta, L}$ and $\mS' = \tuple{S', \delta', L'}$ be \FLTSs and let $x \in S$ and $x' \in S'$. We write $x \lesssim^{ds} x'$ to denote that there exists a directed simulation $Z$ between $\mS$ and $\mSp$ such that $Z(x,x')$ holds. We also write \mbox{$x \leq^{\mathit{pos}} x'$} (resp.\ \mbox{$x \leq^{\mathit{pos}}_K x'$}) to denote that $\varphi^\mS(x) \leq \varphi^\mSp(x')$ for all formulas $\varphi$ of~\fpdPDL (resp.\ \fpdK). The following corollary follows immediately from Theorems~\ref{theorem: dir-sim HM} and~\ref{theorem: pres-dir-sim}. 

\begin{corollary} \label{cor: CDS}
Let $\mS = \tuple{S, \delta, L}$ and $\mS' = \tuple{S', \delta', L'}$ be image-finite \FLTSs and let $x \in S$ and $x' \in S'$. 
\begin{itemize}
\item Then, 
\longer{\[ x \lesssim^{ds} x'\ \ \textrm{iff}\ \ x \leq^{\mathit{pos}}_K x', \]}
\shorter{$x \lesssim^{ds} x'$ iff $x \leq^{\mathit{pos}}_K x'$,}
and therefore whether $x \leq^{\mathit{pos}}_K x'$ or not does not depend on the used t-norm $\fand$. 

\item If $\mS$ and $\mSp$ are witnessed w.r.t.~\fpdPDL, then
\longer{\[ x \leq^{\mathit{pos}} x'\ \ \textrm{iff}\ \ x \lesssim^{ds} x'\ \ \textrm{iff}\ \ x \leq^{\mathit{pos}}_K x', \]}
\shorter{\mbox{$x \leq^{\mathit{pos}} x'$} iff $x \lesssim^{ds} x'$ iff $x \leq^{\mathit{pos}}_K x'$,}
and therefore whether \mbox{$x \leq^{\mathit{pos}} x'$} or not does not depend on the used t-norm~$\fand$. 
\end{itemize}
\end{corollary}


\section{Conclusions}
\label{section: conc}

Simulation and directed simulation are useful notions for comparing observational behaviors of automata and LTSs. Before the current work, there was the lack of logical characterizations of crisp simulations between fuzzy structures w.r.t.\ fuzzy modal logics (or their variants) that use a residuated lattice or a t-norm-based semantics. Furthermore, logical characterizations of crisp directed simulations for fuzzy structures had not been studied. 

In this paper, we have provided and proved logical characterizations of crisp simulations and crisp directed simulations between \FLTSs w.r.t.\ fragments of the fuzzy modal logic \fdPDL under a general t-norm-based semantics. The preservation result for crisp simulations (resp.\ crisp directed simulations) has been formulated for \fedPDL (resp.\ \fpdPDL), whereas the Hennessy-Milner property for them has been formulated for a minimized fragment of \fedPDL (resp.\ \fpdPDL) in order to increase the generality.  


\bibliography{BSfDL}
\bibliographystyle{IEEEtran}

\end{document}